\documentclass[sigconf]{acmart}

\usepackage{algorithm}
\usepackage{algpseudocode}
\usepackage{amsthm}

\def\*#1{\mathbf{#1}}
\def\~#1{\boldsymbol{#1}}

\AtBeginDocument{%
  }

\setcopyright{acmlicensed}
\copyrightyear{2018}
\acmYear{2018}
\acmDOI{}

\acmConference[Conference acronym 'XX]{Make sure to enter the correct
  conference title from your rights confirmation emai}{June 03--05,
  2018}{Woodstock, NY}
\acmISBN{978-1-4503-XXXX-X/18/06}

\usepackage{algorithm}
\usepackage{algpseudocode}
\usepackage{amsthm}
\usepackage{amsmath}
\usepackage{graphicx}
\usepackage{caption}
\usepackage{subcaption}
\theoremstyle{plain}
\newtheorem{theorem}{Theorem}[section]
\newtheorem{definition}[theorem]{Definition}
\newtheorem{proposition}[theorem]{Proposition}
\newtheorem{example}[theorem]{Example}
\newtheorem{remark}[theorem]{Remark}
\def\*#1{\mathbf{#1}}
\def\~#1{\boldsymbol{#1}}

\begin{document}

\title{Why Groups Matter: Necessity of Group Structures in Attributions}

\author{Dangxing Chen, Jingfeng Chen, and Weicheng Ye}
\email{dangxing.chen@dukekunshan.edu.cn}
\affiliation{%
  \institution{Duke Kunshan University}
  \country{China}
}








\renewcommand{\shortauthors}{Trovato et al.}

\begin{abstract}
  Explainable machine learning methods have been accompanied by substantial development. Despite their success, the existing approaches focus more on the general framework with no prior domain expertise. High-stakes financial sectors have extensive domain knowledge of the features. Hence, it is expected that explanations of models will be consistent with domain knowledge to ensure conceptual soundness. 
  In this work, we study the group structures of features that are naturally formed in the financial dataset. Our study shows the importance of considering group structures that conform to the regulations. When group structures are present, direct applications of explainable machine learning methods, such as Shapley values and Integrated Gradients, may not provide consistent explanations; alternatively, group versions of the Shapley value can provide consistent explanations. We contain detailed examples to concentrate on the practical perspective of our framework.
\end{abstract}





\maketitle

\section{Introduction}

Compared to traditional methods, machine learning (ML) models often increase accuracy at the expense of black-box functionality. 
The model explanation is extremely substantial for highly regulated industries such as finance, as stressed in the model risk management handbook by the Office of the Comptroller of the Currency (OCC) \cite{OCC2021model}. Recently, the Consumer Financial Protection Bureau (CFPB) confirmed that anti-discrimination law requires companies to provide detailed explanations when denying a credit application from clients when using ML methods for decision-making \footnote{https://www.consumerfinance.gov/about-us/newsroom/cfpb-acts-to-protect-the-public-from-black-box-credit-models-using-complex-algorithms/}. In response to the growing requirements from the regulators, researchers and especially practitioners are investigating explainable ML methods to provide the required interpretability from the regulation perspective.

Explainable ML methods have been successful in the past. Among these, axiom-based attribution methods, such as Shapley value \cite{sundararajan2020many,lundberg2017unified} and Integrated Gradients (IG) \cite{sundararajan2017axiomatic}, provide both mathematical rigor and practical explanations. An attribution method involves assigning the prediction score of a model based on its base features. An attribution to a base feature can be viewed as the degree to which the feature contributes to the prediction. The allocation of attributes is determined by preserving the desired axioms for fairness \cite{sundararajan2020many}.

Despite the success of attribution methods, the existing analysis has primarily focused on general axioms without prior knowledge, such as completeness, linearity, dummy, and symmetry \cite{sundararajan2020many}. These axioms may be sufficient to provide reasonable explanations in many traditional ML applications, such as computer vision and large language models. However, financial practice applications have benefited from extensive domain knowledge such as monotonicity \cite{chen2023address} and diminishing marginal effects \cite{gupta2020multidimensional}. We must ensure that decision-making models make reasonable predictions grounded in these domains. Accordingly, the OCC's model risk management handbook \cite{OCC2021model} emphasizes conceptual soundness to reflect sound theory and business practice. 
The same requirement can be found in financial models where we should provide a reasonable explanation that conforms to financial theory. Therefore, attribution methods based on financial knowledge have received increasing attention from practitioners and researchers \cite{chen2024asset,shalit2021shapley}. 

This paper focuses on the group structures of features.
Historically, economic studies have taken group structures into account \cite{owen1977values,calvo2013shapley,kamijo2009two}, such as geographic location and political party affiliation. 
Group structures are also common in finance, for example, several features may be used to describe the same characteristic, such as past-due payments in credit scoring with a variety of durations. Despite this, group structures in finance have been neglected. Recently, there have been calls for group explanations \cite{ijcai2022p778}. 

In this paper, we address the following question. \textbf{Is it necessary to take into account group structures when explain ML models?} 
To answer this question, we provide several group-based axioms that describe what we can expect in situations where there is a natural group structure of features based on a specific domain. The analysis of axioms demonstrates that popular attribution methods, such as Shapley value and IG, may not be able to preserve these axioms, thus providing unsatisfactory explanations. Therefore, it is important to exercise caution in ML models when group structure is involved. Alternatively, the group version of Shapley values preserves key axioms while providing consistent explanations.
In summary, our contributions to this work are:
\begin{enumerate}
    \item we propose group axioms for attribution methods with a rigorous theoretical framework.
    \item we provide detailed explanations when group axioms are needed and why they are important. 
\end{enumerate}

\section{Preliminaries}
\label{prerequisites}
For problem setup, assume we have $\mathcal{D} \times \mathcal{Y}$, where $\mathcal{D}$ is the dataset with $n$ samples and $m$ features and $\mathcal{Y}$ is the corresponding numerical values in regression and labels in classification. We denote a class of functions $f: \mathbb{R}^m \rightarrow \mathbb{R}$ by $\mathcal{F}$.  For simplicity, we assume $\*x \in \mathbb{R}^m$ and $f$ is differentiable almost everywhere.



\subsection{Baseline Attribution Methods}
Following \cite{lundstrom2022rigorous}, we denote the point of interest $\overline{\*x}$ to explain as an explicand,  $\*x'$ a baseline, and $\*x$ the general input. The Baseline Attribution Method that interprets features' importance is defined.

\begin{definition}[Baseline Attribution Method (BAM)] \label{def:BAM}
    Given $\overline{\*x}, \*x' \in \mathbb{R}^m$, $f \in \mathcal{F}$, a baseline attribution method is any function of the form $\~{\mathcal{A}}(\overline{\*x},\*x',f): \mathbb{R}^m \times \mathbb{R}^m \times \mathcal{F} \rightarrow \mathbb{R}^m$. We denote $\mathcal{A}_i(\overline{\*x},\*x',f)$ or sometimes simply $\mathcal{A}_i(f)$ or $\mathcal{A}_i$ as the $i$-th attribution of $\~{\mathcal{A}}(\overline{\*x},\*x',f)$. 
\end{definition}

\subsection{Shapley Value and Integrated Gradients}

\subsubsection{(Baseline) Shapley Value}

The Shapley value \cite{lundberg2017unified} takes as input a set function $v:2^M \rightarrow \mathbb{R}$, where $M = \{1, \dots, m\}$. The Shapley value produces attributions $\text{SH}_i$ for each player $i \in M$ by
\begin{align}
    \text{SH}_i = \sum_{S \subseteq M \backslash i} \frac{|S|! (|M|-|S|-1)!}{M!} (v(S \cup i) - v(S)).
\end{align}
Here, we focus on the Baseline Shapley (BShap), which calculates
\begin{align}
    v(\overline{\*x},\*x',f;S) = f(\overline{\*x}_S; \*x'_{M \backslash S}).
\end{align}
That is, baseline values replace the feature's absence. For example, suppose $f(x_1,x_2) = x_1+x_2$, $\overline{\*x} = (\overline{x}_1, \overline{x}_2)$, $\*x' = (0,0)$, and $S=\{1\}$, then we have $v(\overline{\*x},\*x',f;S) = f(\overline{x}_1,0)$. We denote the $i$-th attribution of BShap by $\text{BS}_i(\overline{\*x},\*x',f)$ and $\text{BS}_i(f)$ or $\text{BS}_i$ sometimes. We focus on BShap since it has better theoretical properties by preserving desired axioms, as discussed in \cite{sundararajan2020many}.

\subsubsection{Integrated Gradients}
Given $\overline{\*x}, \*x'$, and $f$,  the $i$-th component of $\overline{\*x}$ of Integrated Gradients (IG) \cite{sundararajan2017axiomatic} is calculated by
    \begin{align}
        \text{IG}_i(\overline{\*x},\*x',f) = (\overline{x}_i-x_i') \int_0^1 \frac{\partial f}{\partial x_i} \left( \*x' + t(\overline{\*x}-\*x') \right) \ dt.
    \end{align}
For simplicity, we often use $\text{IG}_i(f)$ or $\text{IG}_i$ for $\text{IG}_i(\overline{\*x},\*x',f)$.

\subsection{Individual and Pairwise Monotonicity}

Two types of monotonicity that are commonly used in practice are discussed here \cite{chen2023address,gupta2020multidimensional}. Without loss of generality (WLOG), we assume that all monotonic features are monotonically increasing throughout the paper. Individual monotonicity \cite{liu2020certified,runje2023constrained}, as one of the most commonly used domain knowledge, is defined below. 
\begin{definition}[Individual Monotonicity] \label{def:indi_mono}
Suppose we partition the input $\*x$ into $\*x = (x_{i}, \*x_{\neg})$. 
We say $f$ is individually monotonic with respect to $x_{i}$ if $\forall c >0$,
$
 f(x_{i}, \*x_{\neg}) \leq f(x_i+c, \*x_{\neg}).
$
\end{definition}

\begin{example}\label{eg:IM}
    In Credit Scoring, the probability of default should be individually monotonic with respect to the number of past-due payments. 
\end{example}

In practice, certain features are intrinsically more important than others. It is referred to as pairwise monotonicity and has received increasing attention in recent years \cite{chen2023address,chen2022monotonic,gupta2020multidimensional}.


\begin{definition}[Strong Pairwise Monotonicity] \label{def:strong_mono}
Suppose we partition $\*x = (x_{i},x_{j},\*x_{\neg})$. WLOG, we assume that $x_{i}$ has greater significance than $x_{j}$. We say $f$ is strongly monotonic with respect to $x_{i}$ over $x_{j}$ if $\forall \*x$ and $c>0$,
$
     f(x_{i},x_{j}+c,\*x_{\neg}) \leq f(x_{i}+c,x_{j},\*x_{\neg}).
$
\end{definition}

\begin{example}\label{eg:SPM}
 In credit scoring, suppose $f$ calculates the probability of default, $x_{i}$ counts the number of past dues that have been outstanding for more than three months, and $x_{j}$ counts the number of past dues that are outstanding less than three months. As $x_{i}$ is always more important than $x_{j}$, $f$ should be strongly monotonic with respect to $x_{i}$ over $x_{j}$. The pairwise relationship requires that whenever there is an additional past due, a longer past due is always more serious.         
\end{example}

\subsection{Axioms}

Here, we provide important axioms considered for BAMs. 

\begin{enumerate}
    \item Completeness: $\sum_{i=1}^m \mathcal{A}_i(\overline{\*x},\*x',f) = f(\overline{\*x}) - f(\*x')$.
    \item Linearity: $\forall \alpha, \beta \in \mathbb{R}$, $\mathcal{A}_i(\alpha f + \beta g) = \alpha \mathcal{A}_i(f) + \beta \mathcal{A}_i (g)$.
    \item Dummy: If $\partial_i f \equiv 0$, then $\mathcal{A}_i(\overline{\*x},\*x',f) = 0$. 
    \item Symmetry: For a given $i,j$, define $\*x^*$ by swapping the values of $x_i$ and $x_j$, Now suppose that $\forall \*x$, $f(\*x) = f(\*x^*)$. Then if $\overline{x}_i=\overline{x}_j$ and $x_i'=x_j'$, we have $\mathcal{A}_i(\overline{\*x},\*x',f) = \mathcal{A}_j(\overline{\*x},\*x',f)$.
    \item Affine Scale Invariance (ASI): Suppose we partition $\*x = (x_i,\*x_{\neg})$. Define affine transformation as 
    \begin{align*}
        \*h(\*x;i,c,d) = \*h((x_i,\*x_{\neg});i,c,d) = (cx_i+d,\*x_{\neg}).
    \end{align*}
    For simplicity, we denote $\*h(\*x;i,c,d)$ by $\*h_i(\*x)$ for short. For any index $i$, $c \neq 0$ and $d \in \mathbb{R}$, 
    \begin{align*}
        \mathcal{A}_i(\overline{\*x},\*x', f) = \mathcal{A}_i ( \*h_i(\overline{\*x}),\*h_i(\*x'), f \circ \*h_i^{-1}(\*x) ).
    \end{align*}
    \item Demand Individual Monotonicity (DIM): If we partition $\*x = (x_i,\*x_{\neg})$ and $f$ is individually monotonic with respect to $x_i$, then $\mathcal{A}_i((\overline{x}_i+c,\overline{\*x}_{\neg}),\*x',f) \geq \mathcal{A}_i((\overline{x}_i,\overline{\*x}_{\neg}),\*x',f)$, $\forall c > 0$. 
\end{enumerate}

\begin{theorem}\cite{lundstrom2022rigorous,sundararajan2020many,sundararajan2017axiomatic}
    BShap preserves all axioms (1)-(6). IG preserves preserves (1)-(5).
\end{theorem}

\section{Group Structure}

In many applications, features should naturally be classified. Moreover, features can be combined to form groups that share similar characteristics.  Mathematically, we consider a coalition structure over $M = \{1, \dots, m\}$ is a partition of $M$, that is, $B = \{B_1, \dots, B_l\}$ is a group structure if $\cup_{1 \leq i \leq l} B_i = M$ and $B_i \cap B_j = \emptyset$ if $i \neq j$. We also assume $B_i \neq \emptyset$ for all $i$. Denote $\mathcal{B}(M)$ as the set of all coalition structures over $M$. Here are some examples of different practices:
\begin{itemize}
    \item Credit Scoring \cite{chen2023address,chen2022monotonic}: past-due payment information includes the number of past-due payments in different durations. For instance, features that count the number of past dues less than three months and more than three months are naturally in the same group. 
    \item Auto Insurance \footnote{\url{https://www.kaggle.com/datasets/sagnik1511/car-insurance-data/discussion}}: auto insurance depends on the number of past accidents. The accidents could be further divided based on the injury of autos and passengers. Features with different seriousness can be put into the same group. 
    \item Fraud Detection \cite{barata2021active}: to determine if there are any abnormal activities, transaction frequency is an important indicator. It is possible, for instance, to count frequencies daily, such as today, yesterday, etc. This group of features measures the frequency of transactions in an account. 
\end{itemize}

Similar to Definition~\ref{def:BAM}, we define the group attribution method. 

\begin{definition}[Group Attribution Method (GAM)]
    Given $\overline{\*x}, \*x' \in \mathbb{R}^m$, $f \in \mathcal{F}$, and a group structure $B$, a group attribution method is any function of the form $\~{\mathcal{A}}(\overline{\*x},\*x',f): \mathbb{R}^m \times \mathbb{R}^m \times \mathcal{F} \rightarrow \mathbb{R}^l$. We denote $\mathcal{A}_{B_i}(\overline{\*x},\*x',f)$ or sometimes simply $\mathcal{A}_{B_i}(f)$ or $\mathcal{A}_{B_i}$ as the $i$-th group attribution for the group $B_i$ in $\~{\mathcal{A}}(\overline{\*x},\*x',f)$. 
\end{definition}

We are interested in the attribution $\mathcal{A}_{B_i}$ of a group $B_i$. One method is to calculate individual attributions first, and then add them together. Group attributions by Shapley value and IG can be calculated simply by adding the individual attributions together, 
\begin{align*}
    \mathcal{A}_{B_i}(\overline{\*x},\*x',f) = \sum_{j \in B_i} \mathcal{A}_j(\overline{\*x},\*x',f).
\end{align*}
Alternatively, one may decide to determine the group attributions directly, and then study individual attributions, as discussed below.

\subsection{Group Shapley Value}

We review the Shapley value for games with group structures \cite{calvo2013shapley,kamijo2009two,owen1977values}, where the rules governing cooperation among group members differ from those governing interactions among group members. A game $(M,v(\overline{\*x},\*x',f))$ with a group structure $B \in \mathcal{B}(M)$ is denoted by $(B,M,v(\overline{\*x},\*x',f))$. Define $L = \{1, \dots, l\}$ and 
\begin{align}
    v_B(\overline{\*x},\*x',f;T) = v(\overline{\*x},\*x',f;\cup_{i \in T} B_i), \forall T \subseteq L.
\end{align}
Define $(L,v_B(\overline{\*x},\*x',f))$ as the game induced by $(B,M,v(\overline{\*x},\*x',f))$ by considering the unions of $B$ as players. We refer to this as group BShap (GShap) and the $i$-th attribution is calculated as
\begin{align} \label{eq:GBShap}
    \text{GS}_{B_i} (\overline{\*x},\*x',f) &= \sum_{T \subseteq L \backslash i} c(T) \Delta v_B(T;i) \\  
    c(T) &= \frac{|T|!(|L|-|T|-1)!}{|L|!}, \\
    \Delta v_B(T;i) &= v_B(\overline{\*x},\*x',f;T \cup i)-v_B(\overline{\*x},\*x',f;T).
\end{align}
GShap is determined in the same manner as Shapley value. Completeness and linearity are preserved. In addition, it takes into account the group version of the dummy and symmetry. 
\begin{itemize}
    \item Group dummy: If $v_B(T \cup i) = v_B(T) \ \forall T \subseteq L \backslash i$, then $\text{GS}_{B_i}=0$. 
    \item Group symmetry: If $v_B(T \cup i) = v_B(T \cup j) \ \forall T \subseteq L \backslash \{i,j\}$, then $\text{GS}_{B_i} = \text{GS}_{B_j}$.
\end{itemize}
GShap's theoretical properties are discussed in \cite{hu2018new}. An example of a demonstration is provided below.

\begin{example} \label{eg:GShap}

Suppose we have $M =\{1,2,3\}$, $B = \{\{1,2\},3\}$, then the characteristic function $v_B$ only consider $v(\{\emptyset\})$, $v(\{1,2\})$, $v(\{3\})$, and $v(\{1,2,3\})$. GShap calculates that
\begin{align*}
    \text{GS}_{\{1,2\}} &= \frac{v(\{1,2\})-v(\{\emptyset\})}{2} + \frac{v(\{1,2,3\})-v(\{3\})}{2}, \\
    \text{GS}_{\{3\}} &= \frac{v(\{3\})-v(\{\emptyset\})}{2} + \frac{v(\{1,2,3\})-v(\{1,2\})}{2}.
\end{align*}

\end{example}

It is possible to further decomposite group attributions. One of the most popular methods is Owen value \cite{owen1977values}. We will have a review below. There exist other ways for further decomposition \cite{kamijo2009two,calvo2013shapley}. However, the purpose of this work is on attributions of groups. Therefore, we only display Owen value for demonstration purposes and further discussion is beyond the scope. 

\subsubsection{Owen Value}

 For Owen value, each $\Delta v_B(T;i)$ is further decomposed by the Shapley Value by preserving completeness, linearity, dummy, and symmetry within the group. The Owen value is calculated by
\begin{align}
    \text{OW}_j &= \sum_{T \subseteq L \backslash i} \sum_{S \subseteq B_i \backslash j} d(S,T) (v(Q \cup S \cup j) - v(Q \cup S)), \\
    Q &= \cup_{k \in T} B_k, \\
    d(S,T) &= \frac{|T|!(|L|-|T|-1)!}{|L|!} \frac{|S|!(|B_i|-|S|-1)!}{|B_i|!}.
\end{align}

\begin{example} \label{eq:Owen}
    Following Example~\ref{eg:GShap}, we further split $\text{GS}_{\{1,2\}}$ by 
    \begin{align*}
        \text{OW}_1 &= \frac{v(\{1\})-v(\emptyset)}{4} + \frac{v(\{1,2\})-v(\{2\})}{4} \\
        &+ \frac{v(\{1,3\})-v(\{3\})}{4} + \frac{v(\{1,2,3\}) - v(\{2,3\})}{4}, \\
        \text{OW}_2 &= \frac{v(\{2\})-v(\emptyset)}{4} + \frac{v(\{1,2\})-v(\{1\})}{4} \\
        &+ \frac{v(\{2,3\})-v(\{3\})}{4} + \frac{v(\{1,2,3\}) - v(\{1,3\})}{4}.
    \end{align*}
\end{example}

\section{Group Axioms} \label{sec:axioms}

This section discusses what should be expected regarding attributions when group structures are involved. Axioms are naturally motivated by financial knowledge.

\subsection{Group Transformation Invariance}

Concepts of invariance have guided concepts across a variety of scientific fields. Invariance under feature transformations has been extensively studied \cite{kvinge2022ways,mroueh2015learning}. The importance of invariance in practice can be attributed to many factors, including the ability to produce more robust results, to be consistent with human intuition, and to facilitate better generalizations.
Accordingly, we would like to preserve invariance as much as possible when explaining ML methods. This section emphasizes the importance of feature invariance in financial practices.  

\subsubsection{Group Affine Scale Invariance}

The original concept of affine scale invariance \cite{friedman1999three} refers to an affine transformation of a single feature. For instance, if $x_i$ represents annual income in dollars, the result should remain the same if the currency unit of $x_i$ is replaced by another currency instead. Using a group affine transformation of features, we generalize this concept and demonstrate it with the following example. 

\begin{example} \label{eg:GASI}
    In Credit Scoring practice, past-due payments could be split according to their duration.
    The number of past-due payments for less than three months and more than three months can be calculated using $x_1, x_2$.  Alternatively, we can use $\widetilde{x}_1, \widetilde{x}_2$ to record the total number of past-due payments and the number of past-due payments that are more than three months. Ultimately, the result should be the same. Let us assume that there is an additional feature $x_3$ that calculates the yearly income. Consequently, two past-due features may be placed in the same group, i.e. $B_1 = \{1,2\}$. In this case, $B_1$ represents all past-due information and there is an affine relationship such that
    \begin{align*}
        \left[ \begin{matrix} \widetilde{x}_1 \\ \widetilde{x}_2 \\ x_3 \end{matrix} \right] = 
        \*h(\*x;B_1,\*A,\*b) = 
        \left[ \begin{matrix} 1 & 1 & 0 \\ 0 & 1 & 0 \\ 0 & 0 & 1 \end{matrix} \right] 
        \left[ \begin{matrix} x_1 \\ x_2 \\ x_3 \end{matrix} \right] 
        = \left[ \begin{matrix}
    \*A & \begin{matrix} 0 \\ 0 \end{matrix} \\
    \begin{matrix} 0 & 0 \end{matrix} & 1
\end{matrix} \right] \left[ \begin{matrix} x_1 \\ x_2 \\ x_3 \end{matrix} \right] + 
\left[ \begin{matrix}
\begin{array}{c}
  \\[-0.7em]
  \*b \\
  \\[-0.7em]
\end{array} \\
0
\end{matrix} \right]
    \end{align*}
    where 
    \begin{align*}
        \*A = \left[ \begin{matrix} 1 & 1  \\ 0 & 1  \end{matrix} \right], 
        \ \ 
        \*b = \left[ \begin{matrix} 0 \\ 0 \end{matrix} \right].
    \end{align*}
    Note $\*A \in \mathbb{R}^{2 \times 2}$ and $\*b \in \mathbb{R}^2$ since we only act transformation on $B_1$. It is easy to determine the inverse of transformation as
    \begin{align*}
    \left[ \begin{matrix} x_1 \\ x_2 \\ x_3 \end{matrix} \right] =
        \*h^{-1}(\widetilde{\*x};B_1,\*A,\*b) = 
        \left[ \begin{matrix}
            1 & -1 & 0 \\
            0 & 1 & 0 \\
            0 & 0 & 1 
        \end{matrix} \right] 
        \left[ \begin{matrix} \widetilde{x}_1 \\ \widetilde{x}_2 \\ x_3 \end{matrix} \right] = 
        \left[ \begin{matrix} \widetilde{x}_1 - \widetilde{x}_2 \\ \widetilde{x}_2 \\ x_3 \end{matrix} \right].
    \end{align*}
    For simplicity, we use $\*h(\*x)$ or just $\*h$ for short. 
    Now consider a simple logistic regression 
    \begin{align*}
        f(x_1,x_2,x_3) &= \sigma(-10+x_1+2x_2+x_3),
    \end{align*}
    where $\sigma(x) = \frac{e^x}{1+e^x}$. Then by calculation, we have
    \begin{align*}
        g(\widetilde{x}_1,\widetilde{x}_2,x_3) &= f \circ \*h^{-1} = \sigma(-10+\widetilde{x}_1+\widetilde{x}_2+x_3).
    \end{align*}
    It is then easy to verify that
    \begin{align*}
        f(x_1,x_2,x_3) = g(\widetilde{x}_1,\widetilde{x}_2,x_3).
    \end{align*}
    As a result, regardless of how past-due features are documented, the attributions of the group should be the same,
    \begin{align*}
        \mathcal{A}_{B_1}(\overline{\*x},\*x',f) = \mathcal{A}_{B_1}(\*h(\overline{\*x}),\*h(\*x'),g).
    \end{align*}
\end{example}

\begin{definition}[\textbf{Group Affine Scale Invariance  (GASI)}] \label{def:GASI}
Suppose we partition $\*x = (\*x_{B_i},\*x_{\neg})$. Define group affine transformation as
\begin{align*}
    \*h(\*x;B_i,\*A,\*b) = \*h((\*x_{B_i},\*x_{\neg});B_i,\*A,\*b) = (\*A \*x_{B_i} + \*b, \*x_{\neg}).
\end{align*}
As a convenience, we will use $\*h(\*x)$ or $\*h$ for short. 
Then for any group $B_i$, invertible matrix $\*A$ and vector $\*b$ which implies that $\*h$ is invertible, 
\begin{align*}
    \mathcal{A}_{B_i}(\overline{\*x},\*x',f) = \mathcal{A}_{B_i}(\*h(\overline{\*x}),\*h(\*x'),f \circ \*h^{-1}).
\end{align*}
    
\end{definition}

\begin{proposition} \label{prop:GShap_invariant}
    GShap attributions are invariant for any invertible group transformations. 
\end{proposition}

\begin{proposition} \label{prop:GASI_BShap_IG}
    GASI is preserved by IG, but not BShap.
\end{proposition}

\begin{example}
    For BShap, we show this by a counterexample. We consider Example~\ref{eg:GASI} with a specific case
    \begin{align*}
        \overline{\*x}  = (0,40,20), \ \ \*x' = (0,0,0).
    \end{align*}
    For this case, the characteristic function is given by
    \begin{align*}
        & v(\overline{\*x},\*x',f;\{\emptyset\}) \approx 0, v(\overline{\*x},\*x',f; \{1,2,3\}) \approx 1, \\
        & v(\overline{\*x},\*x',f;\{1\}) = 0, v(\overline{\*x},\*x',f;\{2\}) \approx 1, v(\overline{\*x},\*x',f;\{3\}) \approx 1, \\
        & v(\overline{\*x},\*x',f;\{1,2\}) \approx 1, v(\overline{\*x},\*x',f;\{1,3\}) \approx 1, v(\overline{\*x},\*x',f;\{2,3\}) \approx 1.
    \end{align*}
    As a result, from the perspective of $f$, $\overline{x}_1$ is dummy, $\overline{x}_2$ and $\overline{x}_3$ are almost symmetric, we have
$        \text{BS}_{B_1}(\overline{\*x},\*x',f) \approx \frac{1}{2}. 
$
    By group affine transformation, we have 
$
        \*h(\overline{\*x}) = (20,20,20).
$
    From the perspective of $g$, $x_1,x_2,x_3$ are perfectly symmetrical. Thus, 
$
        \text{BS}_{B_1}(\*h(\overline{\*x}),\*h(\*x'),g) \approx \frac{2}{3}. 
$
    It is evident that two different ways of recording features result in significantly different results, which is undesirable. 
    
\end{example}

\begin{remark}
    The GASI provides some insights into how features should be grouped. The GASI is not always a concern. As in Example~\ref{eg:GASI}, the feature that counts the number of past-due payments and the feature that calculates annual income have different units, so we do not need to consider their group affine transformation. 
    When features describe the same characteristics (such as how many past-due payments are, and how many inquiries are), and have the same unit, it is important to consider the alternative representation of features through the group affine transformation. If there are no unique reasonable ways to represent features, group representation is recommended.
\end{remark}

\subsubsection{Group linear fractional transformation}

The Affine transformation is not the only transformation that is used in practice. Ratios are often used in finance in order to facilitate better understanding and provide better results when building models. 

\begin{example} \label{eg:GLFI}
    Consider the company bankruptcy prediction. It is common for such predictions to include some features expressed in the fractional form \cite{altman1968financial}. For instance, suppose there are three features $x_1,x_2,x_3$ that represent total assets, retained earnings, and sales. Then in practice, more likely people will use $x_1,\widetilde{x}_2,x_3$ instead, whereas $\widetilde{x}_2 = \frac{x_2}{x_1}$ calculates the profitability ratio that reflects the age and earning capacity of the company. It is hence natural to consider $B_1 = \{1,2\}$ because $\widetilde{x}_2$ is calculated using $x_1$ and $x_2$. To reflect this, a linear fractional transformation can be defined as follows:
    \begin{align*}
        \left[ \begin{matrix} \widetilde{x}_1 \\ \widetilde{x}_2 \\ x_3 \end{matrix} \right] = \*p(\*x;B_1,\*A,\*B,\*c,\*d) = 
        \left[ \begin{matrix}  x_1 \\ \frac{x_2}{x_1} \\ x_3 \end{matrix} \right],
    \end{align*}
    where 
    \begin{align*}
        \*A = \left[ \begin{matrix} 1 & 0 \\ 0 & 1 \end{matrix} \right], 
        \*B = \left[ \begin{matrix} 0 & 0 \\ 1 & 0 \end{matrix} \right], 
        \*c = \left[ \begin{matrix} 0 \\ 0 \end{matrix} \right], 
        \*d = \left[ \begin{matrix} 1 \\ 0 \end{matrix} \right].
    \end{align*}
    That is, we calculate
    \begin{align*}
        \widetilde{\*x} = \frac{\*A \*x + \*c}{\*B \*x + \*d},
    \end{align*}
    whereas 
    the above division is defined as the vector entry-wise ratio. It is also easy to see the inverse of the transformation
    \begin{align*}
        \left[ \begin{matrix} x_1 \\ x_2 \\ x_3 \end{matrix} \right] = \*p^{-1}(\widetilde{\*x};B_1,\*A,\*B,\*c,\*d) = \left[ \begin{matrix} \widetilde{x}_1 \\ \widetilde{x}_1 \widetilde{x}_2 \\ x_3 \end{matrix} \right].
    \end{align*}
    For convenience, we will use $\*p(\*x)$ or $\*p$ for short. Now consider a simple logistic regression 
    \begin{align*}
        f(x_1,x_2,x_3) = \sigma(-10+3x_1+4x_2+5x_3),
    \end{align*}
    where $\sigma = \frac{e^x}{1+e^x}$. Then by calculation, we have
    \begin{align*}
        g(\widetilde{x}_1,\widetilde{x}_2,x_3) = f \circ \*p^{-1} = \sigma(-10+3\widetilde{x}_1 + 4 \widetilde{x}_1 \widetilde{x}_2 + 5x_3).
    \end{align*}
    In a similar fashion to GASI, we should expect the same attributions in the two cases
    \begin{align*}
        \mathcal{A}_{B_1}(\overline{\*x},\overline{\*x}',f) = \mathcal{A}_{B_1}(\*p(\overline{\*x}),\*p(\overline{\*x}'),g).
    \end{align*}
    
\end{example}

\begin{definition}[\textbf{Group Linear Fractional Invariance  (GLFI)}] \label{def:GLFI}
Suppose we partition $\*x = (\*x_{B_i},\*x_{\neg})$. Define linear fractional transformation as
\begin{align*}
    \*p(\*x;B_i,\*A,\*B,\*c,\*d) = \*p((\*x_{B_i},\*x_{\neg});B_i,\*A,\*B,\*c,\*d) = \left( \frac{\*A\*x_{B_i}+\*c}{\*B\*x_{B_i}+\*d}, \*x_{\neg} \right).
\end{align*}
For simplicity, we denote the transformation by $\*p(\*x)$. 
Then for any group $B_i$ and invertible transformation $\*p(\*x)$, 
\begin{align*}
    \mathcal{A}_{B_i}(\overline{\*x},\*x',f) = \mathcal{A}_{B_i} \left(\*p(\overline{\*x}),\*p(\*x'),f \circ \*p^{-1}(\*x) \right).
\end{align*}
    
\end{definition}

\begin{proposition}
    GShap preserves GLFI, but not BShap or IG. 
\end{proposition}

The invariance by GShap is included in Proposition~\ref{prop:GShap_invariant}. The counterexamples for BShap and IG are provided below.

\begin{example}
    For BShap and IG, we provide a counterexample following Example~\ref{eg:GLFI}. Consider the case
    \begin{align*}
        \overline{\*x} = (5,5,5), \*x' = (0,0,0).
    \end{align*}
    By group linear fractional transformation, we have 
    \begin{align*}
        \*p(\overline{x}) = \left(5, 1, 5 \right).
    \end{align*} 
    By calculations, we have
    \begin{align*}
        & \text{BS}_{B_1}(\overline{\*x},\overline{\*x}',f) \approx 0.66 \neq \text{BS}_{B_1}\left(\*p(\overline{\*x}),\*p(\*x'),f \circ \*p^{-1} \right) \approx 0.50, \\ 
        & \text{IG}_{B_1}(\overline{\*x},\overline{\*x}',f) \approx 0.58 \neq \text{IG}_{B_1}\left(\*p(\overline{\*x}),\*p(\*x'),f \circ \*p^{-1} \right) \approx 0.49.
    \end{align*}
    As a result of BShap and IG, we obtain significantly different group attributions. Based on the results, it appears that if we apply BShap and IG directly, it matters a lot how we calculate ratios of features at the first place. 
\end{example}

\begin{remark}
    Based on the GLFI, how we calculate ratios of features could have a significant impact on the calculation. If we obtain different values for group attributions only because ratios are used, then that would be inconsistent. Therefore, to use ratio features instead, we should group them into the same group. 
\end{remark}

\subsection{Group Monotonicity}

In financial practices, the monotonicity of the features plays a crucial role and should be guaranteed from a conceptual soundness and fairness perspective. 
As a result, monotonic ML models have been studied extensively \cite{chen2023address,chen2022monotonic,gupta2020multidimensional}. However, even when models are strictly monotonic, attribution methods may not be able to capture monotonicity accurately \cite{friedman1999three}. Providing reasonable explanations requires an explanation that can reflect monotonicity.

\subsubsection{Group Demand Individual Monotonicity}
As per the original demand individual monotonicity principle \cite{friedman1999three}, if an individual monotonic feature is increased, the corresponding feature attribution should be increased. As illustrated in the following example, if the monotonic feature belongs to a group, then the group attribute should increase if the monotonic feature increases. 

\begin{example} \label{eg:GDIM}
    Let us assume that we are predicting the probability of a loan defaulting. We consider $x_1$ the number of past-due loan payments, $x_2$ the amount of the past-due payments owed, and $x_3$ the external credit score. In this case, both $x_1$ and $x_2$ concerns regarding the past-due information and the probability of default are individually monotonic for both features. When either of these features is increased, we would expect the group attribution to also increase since the risk regarding past-due payments has increased. 
\end{example}

\begin{definition}[\textbf{Group Demand Individual Monotonicity (GDIM)}]\label{def:GDIM}
    Suppose we partition $\*x = (x_i,\*x_{\neg})$. Suppose $f$ is individually monotonic with respect to $x_{i}$, where $i \in B_j$. We say a BAM preserves group demand individual monotonicity if 
    \begin{align}
        \mathcal{A}_{B_j}((\overline{x}_i+c,\overline{\*x}_{\neg}),\*x',f) \geq \mathcal{A}_{B_j}((\overline{x}_i,\overline{\*x}_{\neg}),\*x',f), \forall \overline{\*x} \text{ s.t. } c>0.
    \end{align}
\end{definition}

\begin{remark}
    The GDIM is a generalization of the DIM. In two extreme cases, if $B_j$ only contains a single feature, then the GDIM reduces to DIM. If $B_j$ consists of all features, then GDIM directly follows from the definition of individual monotonicity. 
\end{remark}

\begin{proposition} \label{prop:GDIM}
    GShap preserves GDIM, but not BShap or IG. 
\end{proposition}

\begin{example} 

    IG does not preserve DIM, therefore GDIM is not preserved as DIM is a special case of GDIM. Due to the change in the integral path, the IG is unable to preserve DIM. 
    
    For BShap, we examine a simple three-dimensional case in more detail. Assume we are interested in $x_1$ and that $f$ is monotonic with respect to $x_1$. There is a group structure that $B_1 = \{1,2\}$. 
    We are interested in how the change to $\overline{x}_1$ would affect the attribution. As a convenience, let us assume that $f$ can be differentiated. Based on our calculations, we have
    \begin{align*}
        \frac{\partial \text{BS}_1}{\partial x_1}(\overline{\*x},\*x',f) &=\frac{1}{3} \frac{\partial f(\overline{x}_1,x_2',x_3')}{\partial x_1} + \frac{1}{3} \frac{\partial f(\overline{x}_1,\overline{x}_2,\overline{x}_3)}{\partial x_1} \\
        &+ \frac{1}{6} \frac{\partial f(\overline{x}_1,\overline{x}_2,x_3')}{\partial x_1}
        + \frac{1}{6} \frac{\partial f(\overline{x}_1,x_2',\overline{x}_3)}{\partial x_1}, \\
        \frac{\partial \text{BS}_2}{\partial x_1}(\overline{\*x},\*x',f) &=
        \frac{1}{3} \frac{\partial f(\overline{x}_1,\overline{x}_2,\overline{x}_3)}{\partial x_1} - \frac{1}{3} \frac{\partial f(\overline{x}_1,x_2',\overline{x}_3)}{\partial x_1} \\
        &+ \frac{1}{6} \frac{\partial f(\overline{x}_1,\overline{x}_2,x_3')}{\partial x_1}
        - \frac{1}{6} \frac{\partial f(\overline{x}_1,x_2',x_3')}{\partial x_1}.
    \end{align*}
    Based on this interpretation, we can say that when we increase $\overline{x}_1$, $\text{BS}_1$ also increases. If $\frac{\partial^2 f(\*x)}{\partial x_1 x_2} \geq 0$, then $\frac{\partial \text{BS}_2}{\partial x_1} \geq 0$. Our intuition supports these interpretations. However, if we combine the results, we arrive at an unfavorable result
    \begin{align*}
        \frac{\partial (\text{BS}_1+\text{BS}_2)(\overline{\*x},\*x',f)}{\partial x_1} &= \frac{2}{3} \frac{\partial f(\overline{x}_1,\overline{x}_2,\overline{x}_3)}{\partial x_1} + \frac{1}{3} \frac{\partial f(\overline{x}_1,\overline{x}_2,x_3')}{\partial x_1} \\
        &+ \frac{1}{6} \frac{\partial f(\overline{x}_1,x_2',x_3')}{\partial x_1} - \frac{1}{6} \frac{\partial f(\overline{x}_1,x_2',\overline{x}_3)}{\partial x_1}.
    \end{align*}
    As a result, the response of the group attribution to the change in $\overline{x}_1$ would now depend on $\frac{\partial^2 f(\*x)}{\partial x_1 \partial x_3}$. If we have a large $\frac{\partial^2 f(\*x)}{\partial x_1 \partial x_3}$, then an increase of $\overline{x}_1$ may lead to a decrease of $\text{BS}_1 + \text{BS}_2$, which is counterintuitive. 

    Suppose $z=g(x_1,x_2)$ where $f$ has individual monotonicity with respect to $g$, and $g$ has individual monotonicity with respect to $x_1$. It is therefore natural to expect an increase in $\overline{x}_1$ to increase $\text{BS}_1+\text{BS}_2$. Specifically,  consider a simple three-dimensional example
    \begin{align*}
        f(x_1,x_2,x_3) = \sigma(-10 + 10x_1 + 10x_2 - 10x_3),
    \end{align*}
    where $\sigma(x) = \frac{e^x}{1+e^x}$ and $\*x' = (0,0,0)$. Consider two explicands $(2,2,2)$ and $(4,2,2)$. By calculation, we obtain that
    \begin{align*}
        & \text{BS}_{B_1}((2,2,2),\*x',f) \approx 1.33 \geq 
        \text{BS}_{B_1}((4,2,2),\*x',f) \approx 1.16, \\
        & \text{IG}_{B_1}((2,2,2),\*x',f) \approx 2.00 \geq 
        \text{IG}_{B_1}((4,2,2),\*x',f) \approx 1.50        
                                          , \\
        & \text{GS}_{B_1}((2,2,2),\*x',f) \approx 1.00 \leq \text{GS}_{B_1}((4,2,2),\*x',f) \approx 1.00.
    \end{align*}
    In the content of Example~\ref{def:GDIM}, according to BShap and IG, for a candidate with more past-due payments, the overall past-due attributions could be significantly reduced, which is absurd. 
\end{example}

\begin{remark}
    In the case of GDIM, the situation is more complex. As in Example~\ref{eg:GDIM}, when we consider the number of past-due payments and the amount owing, we would expect to see a group representation since they together describe the payment status, and each additional increase in any feature increases the future payment risk, and therefore should provide a larger group attribution. However, if we consider one feature of past-due payments as well as another yearly income, then an additional past-due payment does not necessarily imply a greater attribution for their combination. 

    GDIM requires an in-depth understanding of domain knowledge. Whenever we expect a change in feature to influence the attributions of a group, these features should be grouped. Typically, this occurs when we use different features to describe the same characteristics, such as delinquency status, inquiry status, and trade frequency. 
\end{remark}

\subsubsection{Group Strong Pairwise Monotonicity}

The motivation of the demand individual monotonicity only applies to changes to a single feature. As shown below, pairwise monotonicity requires consideration of a change for a pair of features.

\begin{example} \label{eg:GSPM}
    Consider Example~\ref{eg:GASI}. Let us assume that a client has a past-due payment of less than three months ($\overline{x}_1$). The client does not pay back within three months, so the payment is more than three months overdue ($\overline{x}_2$). Thus, for this client, $\overline{x}_1$ decreases by one, and $\overline{x}_2$ increases by one. As a result of the increased risk, we would expect the past-due group attributions to increase. 
\end{example}

\begin{definition}[\textbf{Group Strong Pairwise Monotonicity (GSPM)}]
    Suppose $f$ is strongly monotonic with respect to $x_{i}$ over $x_{j}$. Suppose we partition $\*x = (x_{i},x_{j},\*x_{\neg})$ and $i,j \in B_k$. For an explicand $\overline{\*x}$ with $\overline{x}_i > x_i'$ and $\overline{x}_j>x_j'$. Then we say a BAM preserves group average strong pairwise monotonicity if $\forall \overline{\*x}$ and $c>0$, 
    \begin{align}
        \mathcal{A}_{B_k}((\overline{x}_{i}, \overline{x}_{j}+c,\overline{\*x}_{\neg}),\*x',f) \leq \mathcal{A}_{B_k}((\overline{x}_{i}+c, \overline{x}_{j},\overline{\*x}_{\neg}),\*x',f).
    \end{align}
\end{definition}

\begin{proposition} \label{prop:GSPM}
    GShap preserves GSPM, but not BShap or IG.
\end{proposition}

\begin{example} \label{eq:GSPM_fail}
    Following Example~\ref{eg:GSPM}, consider a simple three-dimensional example
    \begin{align*}
        f(x_1,x_2,x_3) = \sigma(-10 + 5x_1 + 10x_2 - 10x_3),
    \end{align*}
    where $\sigma(x) = \frac{e^x}{1+e^x}$ and $\*x' = (0,0,0)$. Consider the group partition that $B_1 = \{1,2\}$ and $B_2 = \{3\}$ and in $B_1$, $f$ is strongly monotonic with respect to $x_2$ over $x_1$. Consider two explicands $(3,3,2)$ and $(0,6,2)$. As a result of longer past-due payments, $(0,6,2)$ carries a greater risk. As a result, the group attributions for past-due amounts should be greater. The calculation, however, reveals that
    \begin{align*}
        & \text{BS}_{B_1}((3,3,2),\*x',f) \approx 1.25 \geq 
        \text{BS}_{B_1}((0,6,2),\*x',f) \approx 1.00, \\
        & \text{IG}_{B_1}((3,3,2),\*x',f) \approx 1.80 \geq 
        \text{IG}_{B_1}((0,6,2),\*x',f) \approx 1.50, \\
        & \text{GS}_{B_1}((3,3,2),\*x',f) \approx 1.00 \leq \text{GS}_{B_1}((0,6,2),\*x',f) \approx 1.00. 
    \end{align*}
    BShap and IG have decreased the attributions for a client with a poor payment status, which is not reasonable. 
\end{example}

\begin{remark}
    If features are strongly pairwise monotonic, they must be considered in the same group. 
\end{remark}

\subsection{Summary of Axioms}

We summarize the results for the preservation of knowledge-inspired axioms by BShap, IG, and GShap in Table~\ref{tab:summary}. As a result of the preservation of axioms when group structures are present, GShap is preferred over BShap and IG.

\begin{table}[ht]
    \centering
\caption{Preservation of knowledge-inspired axioms by BAMs}
\label{tab:summary}
\begin{tabular}{cccccccccccc}
\toprule
 BAM$\backslash$Axioms & ASI & DIM & GASI & GLFI & GDIM & GSPM \\ \hline
BShap & Yes & Yes & No & No & No & No \\
IG & Yes & No & Yes & No & No & No \\
GShap & Yes & Yes & Yes & Yes & Yes & Yes \\
\bottomrule
\end{tabular}
\end{table}


\section{Empirical Examples}

\subsection{Data Description}

We examine the significance of group structure using a widely-used credit scoring dataset from Kaggle, publicly available at \footnote{https://www.kaggle.com/c/GiveMeSomeCredit/overview}. This dataset includes 10 features serving as explanatory variables. A succinct description is provided below. 

\begin{itemize}
\item $x_1-x_3$: Number of times borrower has been 90 days or more past due, 60-89 days, and 30-59 days, respectively (integer). 
\item $x_4$: Total balance on credit cards and personal lines of credit except for real estate and no installment debt such as car loans divided by the sum of credit limits (percentage).
\item $x_5$: Monthly debt payments, alimony, and living costs divided by monthly gross income (percentage). 
\item $x_6$: Monthly income (real number).
\item $x_7$: Number of open loans and lines of credit  (integer). 
\item $x_8$: Number of mortgage and real estate loans (integer). 
\item $x_{9}$: Number of dependents in the family (integer). 
\item $x_{10}$: Client's age (integer)
\item $y$: Whether it is experienced more than 90 days past due.
\end{itemize}
For simplicity, we consider the zero baselines as $\*x' = \*0$. More details of data and models are provided in Appendix~\ref{sec:model_data}.



\subsection{Group Structures}
In this example, we consider the group structure as 
\begin{align*}
    B = \{B_1,B_2,B_3,B_4,B_5\} = \{\{1,2,3\},\{4\},\{5,6\},\{7,8\},\{9\}\,\{10\}\}.
\end{align*}
The potential influence of axioms on each group is summarized in Table~\ref{tab:group_GSMC}. 
Specifically, we discuss why group structures should be considered for this dataset.
\begin{itemize}
    \item As for $x_1-x_3$, they all describe past-due payments but with different durations. We interpret this as past-due information as a group. GASI affects this group since there are other equivalent representations of these features. Specifically, consider a group affine transformation such that $\widetilde{x}_1=x_1, \widetilde{x}_2 = x_2+x_3$ and $\widetilde{x}_3 = x_1+x_2+x_3$. New features are interpreted as the number of past-due payments over 90 days, 60 days, and 30 days past due. There is no difference between the two representations, and we should expect the same group attribution for both. GDIM also affects this group, as any additional past-due payment should result in a greater attribution of past-due status, regardless of whether the payment is 30, 60, or 90 days overdue. Additionally, it is affected by GSPM, whenever the past-due payment is further delayed, there should be a larger attribution of past-due status. Consequently, a worse past-due payment situation, either an additional one or a worse one, should entail a larger attribution to the group. 
    \item $x_5-x_6$ represents the overall information regarding income. In $x_5$, the ratio of some monthly payments to the total income is calculated. This ratio indicates the relative ability of the customer to repay the loan, and this is often considered to be more informative in the finance field. Because of the ratio representation, the group is affected by GLFI. Mathematically, plain payments are equivalent to ratios, and we should expect the group attributions to be the same for both representations. Furthermore, it is also affected by the GDIM: if there is a decrease in monthly income or a rise in the ratio for $x_5$, we would consider a less stable repayment ability, therefore resulting in a large attribution.
    \item In $x_7-x_8$, we describe the overall loan information. There are different types of loans and people believe they have different impacts. Since all features in this group describe the number of loans, the group is affected by GASI. As a possible example, a record could be made of the number of open and closed loans or the number of open loans and the total number of loans. 
\end{itemize}

Due to the lack of natural group structures, $B_2$, $B_5$, and $B_6$ contain only single features.

\begin{table}[ht]
    \centering
\caption{ Group structure in the GSMC }
\label{tab:group_GSMC}
\begin{tabular}{cccccccccccc}
\toprule
 Group & Description & Related Axioms \\ \hline
 $B_1$ & Past-due information & GASI, GDIM, GSPM \\
 $B_2$ & Balance information & GDIM \\ 
 $B_3$ & Income information & GLFI, GDIM \\
 $B_4$ & Loan information & GASI \\
 $B_5$ & Number of dependents & GDIM \\ 
 $B_6$ & Age & \\
\bottomrule
\end{tabular}
\end{table}

\subsection{An Example of Explanations}

For demonstration, we compare three attributions, BShap, IG, and Owen value as a popular choice of GShap. We consider the explicand
\begin{align*}
    \overline{\*x}_1 &= \left[ \begin{matrix} 
    \{2 & 2 & 5\} & \{1.01\} & \{0.57 & 4\} & \{11 & 0\} & \{4\} & \{30\} 
    \end{matrix} \right],
\end{align*}
where we use braces to indicate group structures in Owen value. The corresponding attributions are plotted in Figure~\ref{fig:attr}. From all methods, $x_4$ dominates because the ratio between credit card balances and limits is extremely high, indicating a high level of risk. Additionally, all past-due features have high attributions because many past-due payments are associated with this explanation. Quantitative differences among the three methods can be observed, suggesting that the choice of attribution method matters in this context. 

\begin{figure}[ht]
\centering
\begin{subfigure}{.25\textwidth}
    \includegraphics[width=1\linewidth]{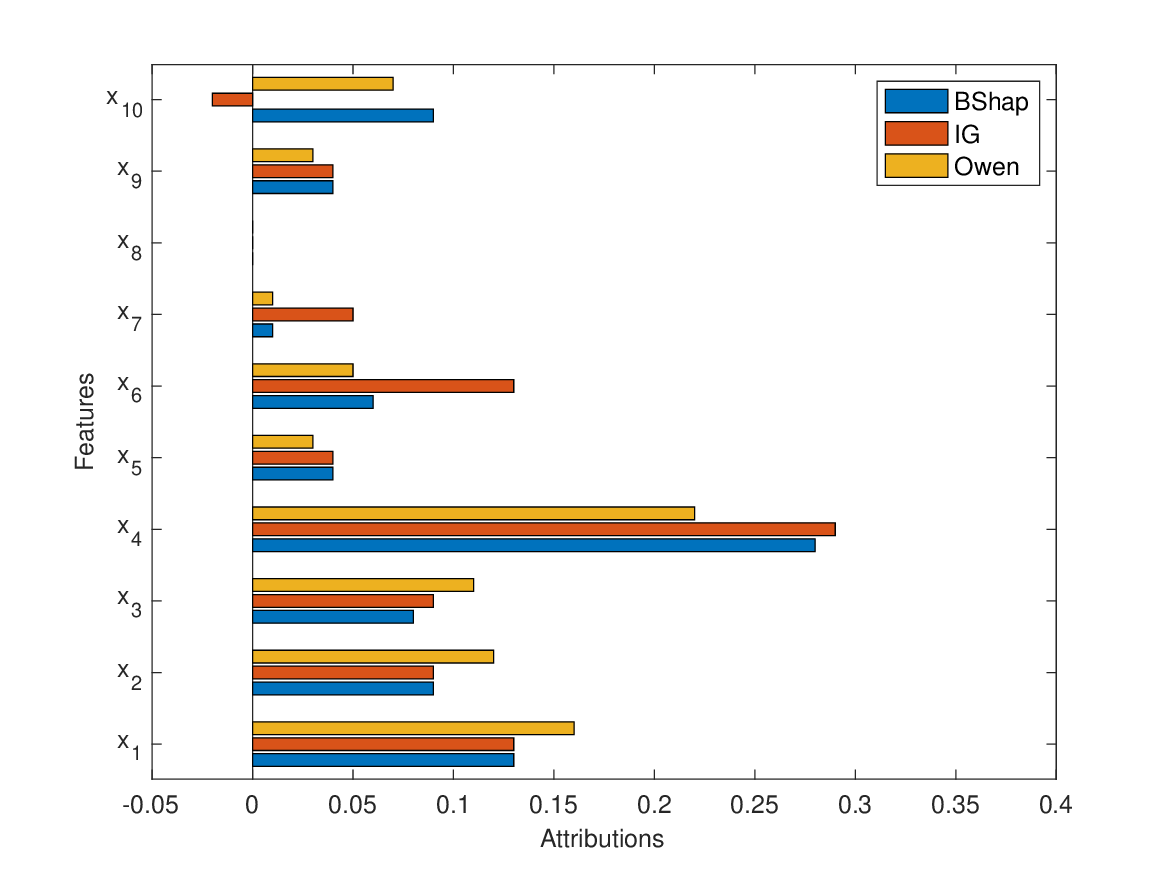}
    \caption{Attributions of $\overline{\*x}_1$}
    \label{fig:attr}
\end{subfigure}%
\begin{subfigure}{.25\textwidth}
    \includegraphics[width=1\linewidth]{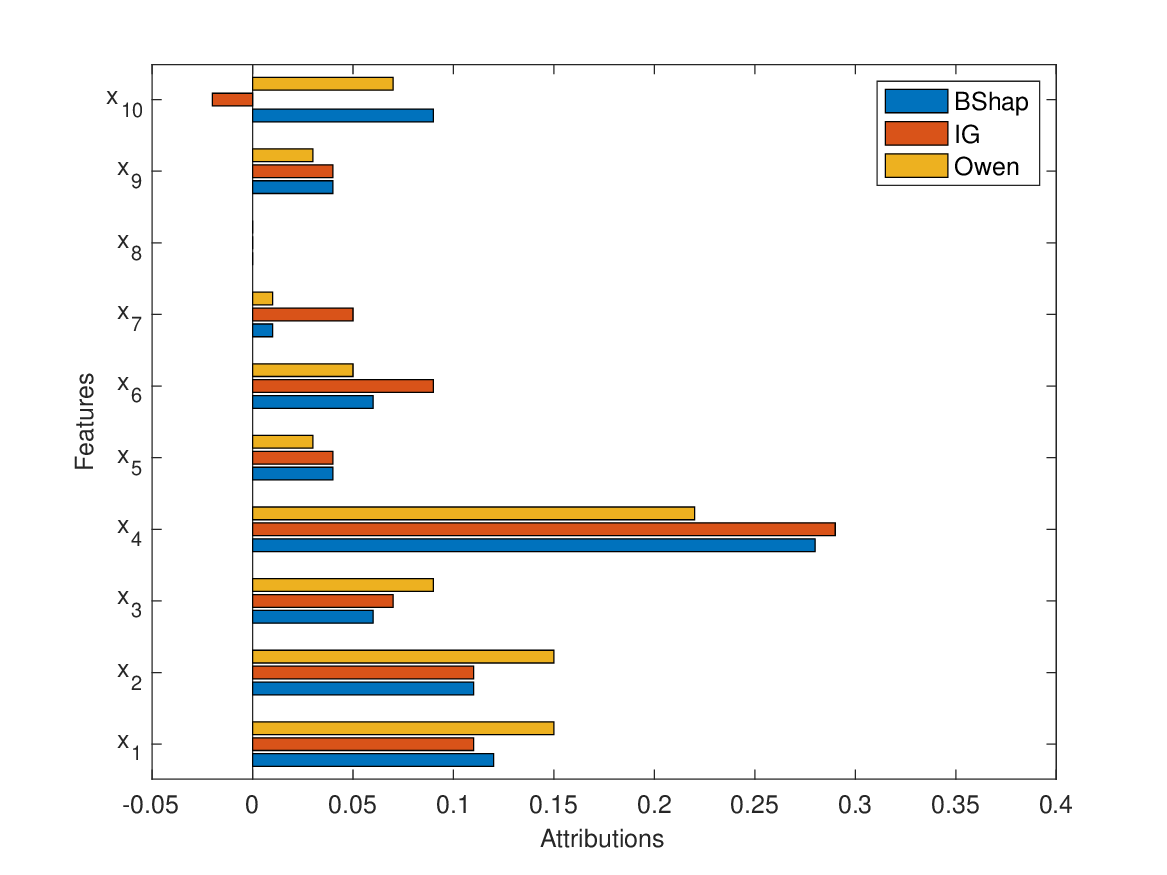}
    \caption{Attributions of $\overline{\*x}_2$}
    \label{fig:attr2}
\end{subfigure}
\begin{subfigure}{.25\textwidth}
    \includegraphics[width=1\linewidth]{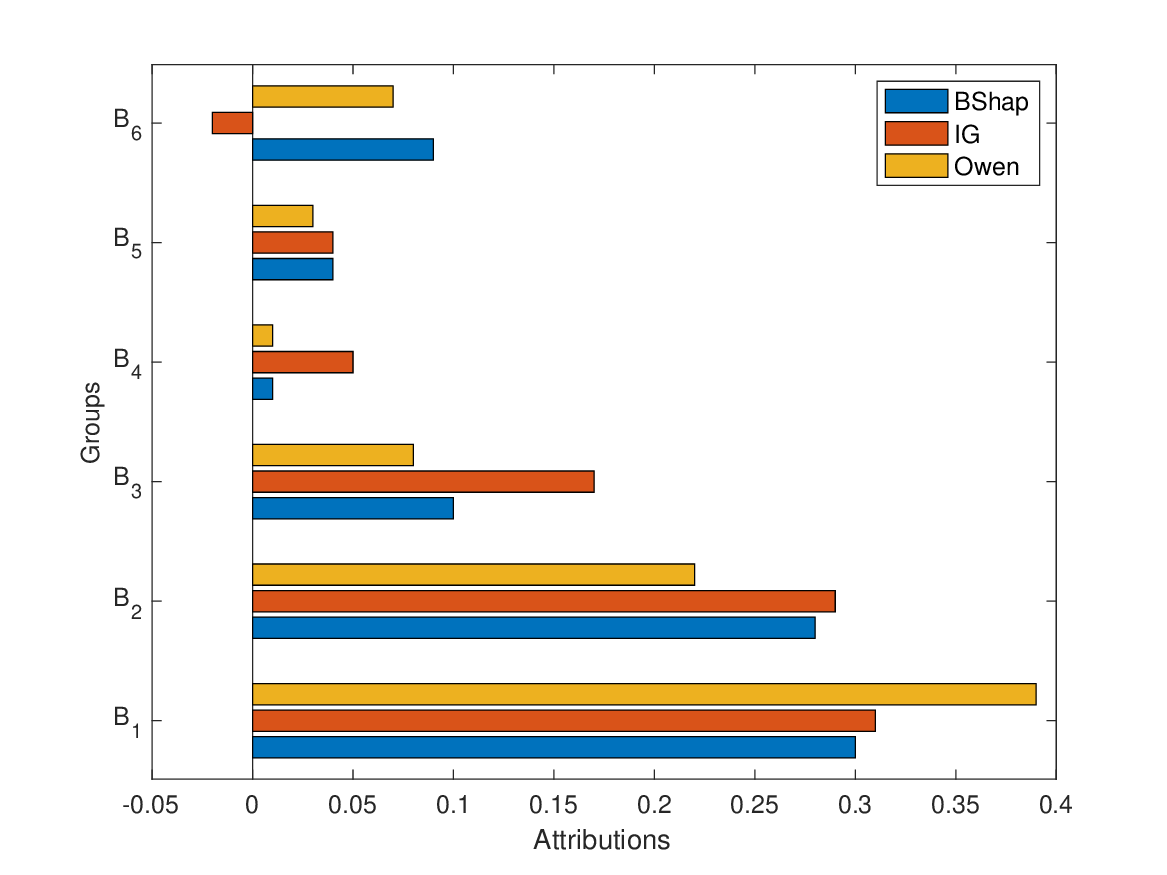}
    \caption{Group attributions of $\overline{\*x}_1$}
    \label{fig:group_attr}
\end{subfigure}%
\begin{subfigure}{.25\textwidth}
    \includegraphics[width=1\linewidth]{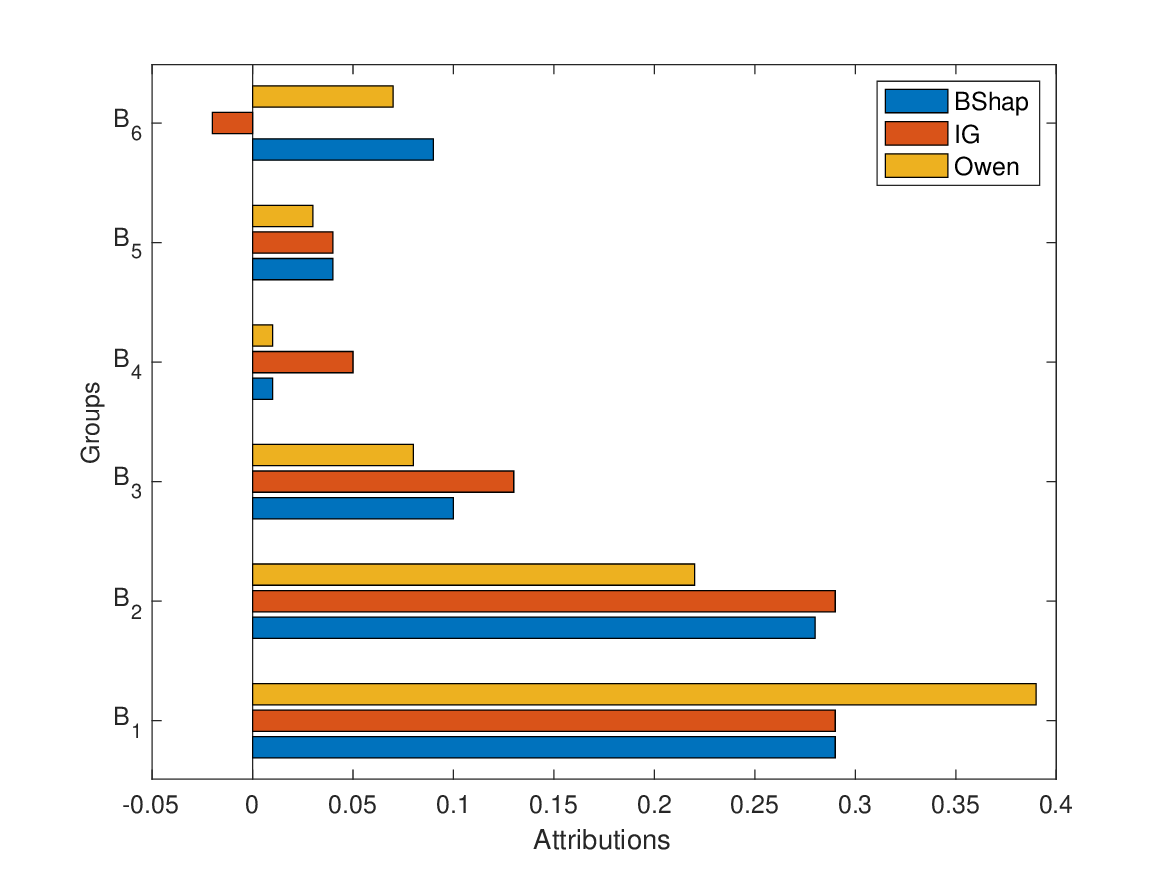}
    \caption{Group attributions of $\overline{\*x}_2$}
    \label{fig:group_attr2}
    \end{subfigure}
\end{figure}

Next, in Figure~\ref{fig:group_attr}, we compare group attributions based on different methods. On a qualitative level, all methods suggest that $B_1$, past-due information, is the dominant group due to a large number of past-due payments for this explicand. This is followed by $B_2$, balance information, and $B_3$, income information. Note that the rank of features and groups will differ for individual and group attributions. This is also expected since the overall past-due payment attribution would constitute the sum of each past-due payment information. Group attributions explain information at a higher level, summarizing information to provide an alternative intuitive explanation. Similar to individual attributions, there are some quantitative differences in group attributions between the three methods. As an example, Owen values assign much more weight to $B_1$ than BShap and IG. Hence, the choice of attribution method does matter, and we suggest the use of GShap, such as Owen, since it preserves desired properties based on financial domain knowledge, as discussed in Section~\ref{sec:axioms}. Below is an example that illustrates this point more specifically. 


We consider a further perturbation of the explicand $\overline{\*x}_1$ with
\begin{align*}
\overline{\*x}_2 = \left[ \begin{matrix} 
    \{2 & 3 & 4\} & \{1.01\} & \{0.57 & 4\} & \{11 & 0\} & \{4\} & \{30\} 
    \end{matrix} \right].
\end{align*}
This means that if a 30-day past-due payment does not pay off within two months then becomes 60 days past-due. The corresponding attributions are plotted in Figure~\ref{fig:attr2}. Due to the changes, attributions of $x_2$ have increased and attributions of $x_3$ have decreased by all methods. In terms of qualitative results, all methods produced reasonable explanations for individual attributions, but when it comes to group attributions in Figure~\ref{fig:group_attr2}, the story has changed. It is important to note that group attributions for $B_1$ have decreased for BShap and IG, violating GSPM. As a result of Proposition~\ref{prop:GSPM}, this did not occur for Owen. In this case, BShap and IG appear to be unreasonable since a worse past-due status should result in overall larger past-due attributions. The Owen value has successfully preserved such property. There are other examples that can be used to demonstrate the advantage of Owen value as discussed in Section~\ref{sec:axioms}. The Owen value is therefore recommended when the group structure is present to achieve better theoretical properties.

\bibliographystyle{ACM-Reference-Format}
\bibliography{sample-base}


\begin{thebibliography}{22}


\ifx \showCODEN    \undefined \def \showCODEN     #1{\unskip}     \fi
\ifx \showDOI      \undefined \def \showDOI       #1{#1}\fi
\ifx \showISBNx    \undefined \def \showISBNx     #1{\unskip}     \fi
\ifx \showISBNxiii \undefined \def \showISBNxiii  #1{\unskip}     \fi
\ifx \showISSN     \undefined \def \showISSN      #1{\unskip}     \fi
\ifx \showLCCN     \undefined \def \showLCCN      #1{\unskip}     \fi
\ifx \shownote     \undefined \def \shownote      #1{#1}          \fi
\ifx \showarticletitle \undefined \def \showarticletitle #1{#1}   \fi
\ifx \showURL      \undefined \def \showURL       {\relax}        \fi
\providecommand\bibfield[2]{#2}
\providecommand\bibinfo[2]{#2}
\providecommand\natexlab[1]{#1}
\providecommand\showeprint[2][]{arXiv:#2}

\bibitem[\protect\citeauthoryear{??}{OCC}{2021}]%
        {OCC2021model}
 \bibinfo{year}{2021}\natexlab{}.
\newblock \bibinfo{title}{Model risk management}.
\newblock
\newblock
\urldef\tempurl%
\url{https://www.occ.treas.gov/publications-and-resources/publications/comptrollers-handbook/files/model-risk-management/index-model-risk-management.html}
\showURL{%
\tempurl}


\bibitem[\protect\citeauthoryear{Altman}{Altman}{1968}]%
        {altman1968financial}
\bibfield{author}{\bibinfo{person}{Edward~I Altman}.}
  \bibinfo{year}{1968}\natexlab{}.
\newblock \showarticletitle{Financial ratios, discriminant analysis and the
  prediction of corporate bankruptcy}.
\newblock \bibinfo{journal}{\emph{The journal of finance}}
  \bibinfo{volume}{23}, \bibinfo{number}{4} (\bibinfo{year}{1968}),
  \bibinfo{pages}{589--609}.
\newblock


\bibitem[\protect\citeauthoryear{Barata, Leite, Pacheco, Sampaio,
  Ascens${\~a}$o, and Bizarro}{Barata et~al\mbox{.}}{2021}]%
        {barata2021active}
\bibfield{author}{\bibinfo{person}{Ricardo Barata}, \bibinfo{person}{Miguel
  Leite}, \bibinfo{person}{Ricardo Pacheco}, \bibinfo{person}{Marco~OP
  Sampaio}, \bibinfo{person}{Jo${\~a}$o~Tiago Ascens${\~a}$o}, {and}
  \bibinfo{person}{Pedro Bizarro}.} \bibinfo{year}{2021}\natexlab{}.
\newblock \showarticletitle{Active learning for imbalanced data under cold
  start}. In \bibinfo{booktitle}{\emph{Proceedings of the Second ACM
  International Conference on AI in Finance}}. \bibinfo{pages}{1--9}.
\newblock


\bibitem[\protect\citeauthoryear{Calvo and Guti{\'e}rrez}{Calvo and
  Guti{\'e}rrez}{2013}]%
        {calvo2013shapley}
\bibfield{author}{\bibinfo{person}{Emilio Calvo} {and} \bibinfo{person}{Esther
  Guti{\'e}rrez}.} \bibinfo{year}{2013}\natexlab{}.
\newblock \showarticletitle{The Shapley-solidarity value for games with a
  coalition structure}.
\newblock \bibinfo{journal}{\emph{International game theory review}}
  \bibinfo{volume}{15}, \bibinfo{number}{01} (\bibinfo{year}{2013}),
  \bibinfo{pages}{1350002}.
\newblock


\bibitem[\protect\citeauthoryear{Chen and Gao}{Chen and Gao}{2024}]%
        {chen2024asset}
\bibfield{author}{\bibinfo{person}{Dangxing Chen} {and} \bibinfo{person}{Yuan
  Gao}.} \bibinfo{year}{2024}\natexlab{}.
\newblock \showarticletitle{Attribution Methods in Asset Pricing: Do They
  Account for Risk?}
\newblock \bibinfo{journal}{\emph{arXiv preprint arXiv:2407.08953}}
  (\bibinfo{year}{2024}).
\newblock


\bibitem[\protect\citeauthoryear{Chen and Ye}{Chen and Ye}{2022}]%
        {chen2022monotonic}
\bibfield{author}{\bibinfo{person}{Dangxing Chen} {and}
  \bibinfo{person}{Weicheng Ye}.} \bibinfo{year}{2022}\natexlab{}.
\newblock \showarticletitle{Monotonic Neural Additive Models: Pursuing
  Regulated Machine Learning Models for Credit Scoring}. In
  \bibinfo{booktitle}{\emph{Proceedings of the Third ACM International
  Conference on AI in Finance}}. \bibinfo{pages}{70--78}.
\newblock


\bibitem[\protect\citeauthoryear{Chen and Ye}{Chen and Ye}{2023}]%
        {chen2023address}
\bibfield{author}{\bibinfo{person}{Dangxing Chen} {and}
  \bibinfo{person}{Weicheng Ye}.} \bibinfo{year}{2023}\natexlab{}.
\newblock \showarticletitle{How to address monotonicity for model risk
  management?}. In \bibinfo{booktitle}{\emph{Proceedings of the 40th
  International Conference on Machine Learning}}
  \emph{(\bibinfo{series}{Proceedings of Machine Learning Research})},
  Vol.~\bibinfo{volume}{202}. \bibinfo{publisher}{PMLR},
  \bibinfo{pages}{5282--5295}.
\newblock


\bibitem[\protect\citeauthoryear{Friedman and Moulin}{Friedman and
  Moulin}{1999}]%
        {friedman1999three}
\bibfield{author}{\bibinfo{person}{Eric Friedman} {and} \bibinfo{person}{Herve
  Moulin}.} \bibinfo{year}{1999}\natexlab{}.
\newblock \showarticletitle{Three methods to share joint costs or surplus}.
\newblock \bibinfo{journal}{\emph{Journal of economic Theory}}
  \bibinfo{volume}{87}, \bibinfo{number}{2} (\bibinfo{year}{1999}),
  \bibinfo{pages}{275--312}.
\newblock


\bibitem[\protect\citeauthoryear{Gupta, Louidor, Mangylov, Morioka, Narayan,
  and Zhao}{Gupta et~al\mbox{.}}{2020}]%
        {gupta2020multidimensional}
\bibfield{author}{\bibinfo{person}{Maya Gupta}, \bibinfo{person}{Erez Louidor},
  \bibinfo{person}{Oleksandr Mangylov}, \bibinfo{person}{Nobu Morioka},
  \bibinfo{person}{Taman Narayan}, {and} \bibinfo{person}{Sen Zhao}.}
  \bibinfo{year}{2020}\natexlab{}.
\newblock \showarticletitle{Multidimensional shape constraints}. In
  \bibinfo{booktitle}{\emph{International Conference on Machine Learning}}.
  PMLR, \bibinfo{pages}{3918--3928}.
\newblock


\bibitem[\protect\citeauthoryear{Hu and Li}{Hu and Li}{2018}]%
        {hu2018new}
\bibfield{author}{\bibinfo{person}{Xun-Feng Hu} {and}
  \bibinfo{person}{Deng-Feng Li}.} \bibinfo{year}{2018}\natexlab{}.
\newblock \showarticletitle{A new axiomatization of the Shapley--solidarity
  value for games with a coalition structure}.
\newblock \bibinfo{journal}{\emph{Operations Research Letters}}
  \bibinfo{volume}{46}, \bibinfo{number}{2} (\bibinfo{year}{2018}),
  \bibinfo{pages}{163--167}.
\newblock


\bibitem[\protect\citeauthoryear{Kamijo}{Kamijo}{2009}]%
        {kamijo2009two}
\bibfield{author}{\bibinfo{person}{Yoshio Kamijo}.}
  \bibinfo{year}{2009}\natexlab{}.
\newblock \showarticletitle{A two-step Shapley value for cooperative games with
  coalition structures}.
\newblock \bibinfo{journal}{\emph{International Game Theory Review}}
  \bibinfo{volume}{11}, \bibinfo{number}{02} (\bibinfo{year}{2009}),
  \bibinfo{pages}{207--214}.
\newblock


\bibitem[\protect\citeauthoryear{Kvinge, Emerson, Jorgenson, Vasquez, Doster,
  and Lew}{Kvinge et~al\mbox{.}}{2022}]%
        {kvinge2022ways}
\bibfield{author}{\bibinfo{person}{Henry Kvinge}, \bibinfo{person}{Tegan
  Emerson}, \bibinfo{person}{Grayson Jorgenson}, \bibinfo{person}{Scott
  Vasquez}, \bibinfo{person}{Tim Doster}, {and} \bibinfo{person}{Jesse Lew}.}
  \bibinfo{year}{2022}\natexlab{}.
\newblock \showarticletitle{In what ways are deep neural networks invariant and
  how should we measure this?}
\newblock \bibinfo{journal}{\emph{Advances in Neural Information Processing
  Systems}}  \bibinfo{volume}{35} (\bibinfo{year}{2022}),
  \bibinfo{pages}{32816--32829}.
\newblock


\bibitem[\protect\citeauthoryear{Liu, Han, Zhang, and Liu}{Liu
  et~al\mbox{.}}{2020}]%
        {liu2020certified}
\bibfield{author}{\bibinfo{person}{Xingchao Liu}, \bibinfo{person}{Xing Han},
  \bibinfo{person}{Na Zhang}, {and} \bibinfo{person}{Qiang Liu}.}
  \bibinfo{year}{2020}\natexlab{}.
\newblock \showarticletitle{Certified monotonic neural networks}.
\newblock \bibinfo{journal}{\emph{Advances in Neural Information Processing
  Systems}}  \bibinfo{volume}{33} (\bibinfo{year}{2020}),
  \bibinfo{pages}{15427--15438}.
\newblock


\bibitem[\protect\citeauthoryear{Lundberg and Lee}{Lundberg and Lee}{2017}]%
        {lundberg2017unified}
\bibfield{author}{\bibinfo{person}{Scott~M Lundberg} {and}
  \bibinfo{person}{Su-In Lee}.} \bibinfo{year}{2017}\natexlab{}.
\newblock \showarticletitle{A unified approach to interpreting model
  predictions}.
\newblock \bibinfo{journal}{\emph{Advances in neural information processing
  systems}}  \bibinfo{volume}{30} (\bibinfo{year}{2017}).
\newblock


\bibitem[\protect\citeauthoryear{Lundstrom, Huang, and Razaviyayn}{Lundstrom
  et~al\mbox{.}}{2022}]%
        {lundstrom2022rigorous}
\bibfield{author}{\bibinfo{person}{Daniel~D Lundstrom},
  \bibinfo{person}{Tianjian Huang}, {and} \bibinfo{person}{Meisam Razaviyayn}.}
  \bibinfo{year}{2022}\natexlab{}.
\newblock \showarticletitle{A rigorous study of integrated gradients method and
  extensions to internal neuron attributions}. In
  \bibinfo{booktitle}{\emph{International Conference on Machine Learning}}.
  PMLR, \bibinfo{pages}{14485--14508}.
\newblock


\bibitem[\protect\citeauthoryear{Mroueh, Voinea, and Poggio}{Mroueh
  et~al\mbox{.}}{2015}]%
        {mroueh2015learning}
\bibfield{author}{\bibinfo{person}{Youssef Mroueh}, \bibinfo{person}{Stephen
  Voinea}, {and} \bibinfo{person}{Tomaso~A Poggio}.}
  \bibinfo{year}{2015}\natexlab{}.
\newblock \showarticletitle{Learning with Group Invariant Features: A Kernel
  Perspective.}
\newblock \bibinfo{journal}{\emph{Advances in neural information processing
  systems}}  \bibinfo{volume}{28} (\bibinfo{year}{2015}).
\newblock


\bibitem[\protect\citeauthoryear{Owen}{Owen}{1977}]%
        {owen1977values}
\bibfield{author}{\bibinfo{person}{Guilliermo Owen}.}
  \bibinfo{year}{1977}\natexlab{}.
\newblock \showarticletitle{Values of games with a priori unions}. In
  \bibinfo{booktitle}{\emph{Mathematical economics and game theory: Essays in
  honor of Oskar Morgenstern}}. Springer, \bibinfo{pages}{76--88}.
\newblock


\bibitem[\protect\citeauthoryear{Rozemberczki, Watson, Bayer, Yang, Kiss,
  Nilsson, and Sarkar}{Rozemberczki et~al\mbox{.}}{2022}]%
        {ijcai2022p778}
\bibfield{author}{\bibinfo{person}{Benedek Rozemberczki},
  \bibinfo{person}{Lauren Watson}, \bibinfo{person}{Péter Bayer},
  \bibinfo{person}{Hao-Tsung Yang}, \bibinfo{person}{Olivér Kiss},
  \bibinfo{person}{Sebastian Nilsson}, {and} \bibinfo{person}{Rik Sarkar}.}
  \bibinfo{year}{2022}\natexlab{}.
\newblock \showarticletitle{The Shapley Value in Machine Learning}. In
  \bibinfo{booktitle}{\emph{Proceedings of the Thirty-First International Joint
  Conference on Artificial Intelligence, {IJCAI-22}}},
  \bibfield{editor}{\bibinfo{person}{Lud~De Raedt}} (Ed.).
  \bibinfo{publisher}{International Joint Conferences on Artificial
  Intelligence Organization}, \bibinfo{pages}{5572--5579}.
\newblock
\urldef\tempurl%
\url{https://doi.org/10.24963/ijcai.2022/778}
\showDOI{\tempurl}
\newblock
\shownote{Survey Track.}


\bibitem[\protect\citeauthoryear{Runje and Shankaranarayana}{Runje and
  Shankaranarayana}{2023}]%
        {runje2023constrained}
\bibfield{author}{\bibinfo{person}{Davor Runje} {and}
  \bibinfo{person}{Sharath~M Shankaranarayana}.}
  \bibinfo{year}{2023}\natexlab{}.
\newblock \showarticletitle{Constrained monotonic neural networks}. In
  \bibinfo{booktitle}{\emph{International Conference on Machine Learning}}.
  PMLR, \bibinfo{pages}{29338--29353}.
\newblock


\bibitem[\protect\citeauthoryear{Shalit}{Shalit}{2021}]%
        {shalit2021shapley}
\bibfield{author}{\bibinfo{person}{Haim Shalit}.}
  \bibinfo{year}{2021}\natexlab{}.
\newblock \showarticletitle{The Shapley value decomposition of optimal
  portfolios}.
\newblock \bibinfo{journal}{\emph{Annals of Finance}} \bibinfo{volume}{17},
  \bibinfo{number}{1} (\bibinfo{year}{2021}), \bibinfo{pages}{1--25}.
\newblock


\bibitem[\protect\citeauthoryear{Sundararajan and Najmi}{Sundararajan and
  Najmi}{2020}]%
        {sundararajan2020many}
\bibfield{author}{\bibinfo{person}{Mukund Sundararajan} {and}
  \bibinfo{person}{Amir Najmi}.} \bibinfo{year}{2020}\natexlab{}.
\newblock \showarticletitle{The many Shapley values for model explanation}. In
  \bibinfo{booktitle}{\emph{International conference on machine learning}}.
  PMLR, \bibinfo{pages}{9269--9278}.
\newblock


\bibitem[\protect\citeauthoryear{Sundararajan, Taly, and Yan}{Sundararajan
  et~al\mbox{.}}{2017}]%
        {sundararajan2017axiomatic}
\bibfield{author}{\bibinfo{person}{Mukund Sundararajan}, \bibinfo{person}{Ankur
  Taly}, {and} \bibinfo{person}{Qiqi Yan}.} \bibinfo{year}{2017}\natexlab{}.
\newblock \showarticletitle{Axiomatic attribution for deep networks}. In
  \bibinfo{booktitle}{\emph{International conference on machine learning}}.
  PMLR, \bibinfo{pages}{3319--3328}.
\newblock


\end{thebibliography}

\appendix

\section{Proofs}

\begin{proof} [Proof of \textbf{Proposition}~\ref{prop:GShap_invariant}]
    For any group transformations $\*q(\*x)$, as long as $\*q$ is invertible, we have $f(\*x) = f \circ \*q^{-1} \circ \*q(\*x)$. Therefore, 
    \begin{align*}
        v_B(\overline{\*x},\*x',f;T) = v_B(\*q(\overline{\*x}),\*q(\*x'),f \circ \*q^{-1}; T).
    \end{align*}
\end{proof}

\begin{proof}[Proof of \textbf{Proposition}~\ref{prop:GASI_BShap_IG}]
    
    For IG, suppose $B = \{B_1, \dots, B_l\}$. Since symmetry, WLOG, we focus on $B_1$. 
    For $i \notin B_j$, for $\*h(\*x) = \*h(\*x;B_j,\*A,\*b)$, where $\*h$ is invertible by Definition, 
    \begin{align*}
        \text{IG}_i(\overline{\*x},\*x',f) &= (\overline{x}_i-x_i') \int_0^1 \frac{\partial f}{\partial x_i}(\*x'+(\overline{\*x}-\*x')t) \ dt \\
        &=(\overline{x}_i-x_i') \int_0^1 \frac{ \left( \partial f \circ \*h^{-1} \right)}{\partial x_i} (\*h(\*x'+(\overline{\*x}-\*x')t)) \ dt \\
        &= (\overline{x}_i-x_i') \int_0^1 \frac{ \left( \partial f \circ \*h^{-1} \right)}{\partial x_i} (\*h(\*x')+(\*h(\overline{\*x})-\*h(\*x'))t) \ dt \\
        &= \text{IG}_i(\*h(\overline{\*x}),\*h(\*x'),f \circ \*h^{-1}).
    \end{align*}
    As a result, group affine transformation doesn't affect attributions of features outside of the group. 
    Suppose all features $x_j$ that are not in $B_1$ have been applied the corresponding affine scale transformation $\*h_1(\*x)$. Then for $i \in B_1$, we have
    \begin{align*}
        \text{IG}_i(\overline{\*x},\*x',f) = \text{IG}_i(\*h_1(\overline{\*x}),\*h_1(\*x'),f \circ \*h_1^{-1}). 
    \end{align*}
    By completeness, we have
    \begin{align*}
        \sum_{j \notin B_1} \text{IG}_j(\overline{\*x},\*x',f) = \sum_{j \notin B_1} \text{IG}_j (\*h_1(\overline{\*x}),\*h_1(\*x'),f \circ \*h_1^{-1}).
    \end{align*}
    Now let the entire affine transformation as $h_2(\*x)$, for $i \notin B_1$, we have
    \begin{align*}
        \text{IG}_i(\overline{\*x}, \*x', f) = \text{IG}_i(\*h_2(\overline{\*x}), \*h_2(\*x'), f \circ \*h_2^{-1}). 
    \end{align*}
    By completeness, for $i \in B_1$, we have
    \begin{align*}
        \text{IG}_i(\overline{\*x}, \*x', f) = \text{IG}_i(\*h_2(\overline{\*x}), \*h_2(\*x'), f \circ \*h_2^{-1}). 
    \end{align*}

\end{proof}

\begin{proof}[Proof of \textbf{Proposition}~\ref{prop:GDIM}]
    For GShap, by individual monotonicity, we have
    \begin{align*}
        & \text{GS}_{B_j} ((\overline{x}_i+c,\overline{\*x}_{\neg}),\*x',f) - \text{GS}_{B_j}(\overline{\*x},\*x',f) \\
        =& \sum_{T \subseteq L \backslash j} c(T) (v_B((\overline{x}_i+c,\overline{\*x}_{\neg}),\*x',f;T \cup i)-v_B(\overline{\*x},\*x',f;T \cup i))) 
       \geq 0.
    \end{align*}
\end{proof}

\begin{proof}[Proof of \textbf{Proposition}~\ref{prop:GSPM}]
    For GShap, by strong pairwise monotonicity, we have
    \begin{align*}
        & \text{GS}_{B_k} ((\overline{x}_{i}, \overline{x}_{j}+c,\overline{\*x}_{\neg}),\*x',f)  - \text{GS}_{B_k}((\overline{x}_{i}+c, \overline{x}_{j},\overline{\*x}_{\neg}),\*x',f) \\
        =& \sum_{T \subseteq L \backslash j} c(T) v_B((\overline{x}_{i}, \overline{x}_{j}+c,\overline{\*x}_{\neg}),\*x',f;T \cup i) \\
        -& c(T) v_B((\overline{x}_{i}+c, \overline{x}_{j},\overline{\*x}_{\neg}),\*x',f;T \cup i) \leq 0.
    \end{align*}
\end{proof}

\section{Model and Data} \label{sec:model_data}

For simplicity, data with missing variables are removed. 
Past dues greater than four times are replaced by four due to the rarity. This also applies to $x_9$ if its value exceeds five. We categorize $x_6$ into the following intervals: $[0,\$2500)$, $[\$2,500,\$5,000)$, $[\$5,000,\$7,500)$, $[\$7,500,\$10,000)$, $[\$10,000,\$50,000)$, and $[\$50,000,\infty)$. Afterward, they are transformed from five to zero so that $f$ increases monotonically with respect to $x_6$.  We make such a choice to make features as easy to understand as possible for customers. This is not a unique choice. The model performance has been monitored to ensure that the accuracy does not deteriorate.
When checking for accuracy, the dataset is randomly partitioned into $70\%$ training and $30\%$ test sets.

We use the monotonic groves of neural additive models \cite{chen2023address} since it preserves required individual and pairwise monotonicity for this dataset. This is not a unique choice and other models can also be applied.  
For $x_1-x_3$, we enforce strong pairwise monotonicity. We enforce individual monotonicity for $x_6$ and $x_9$. 
The AUC of the model is around $85\%$, which is consistent with the literature. It might be possible to improve model performance by further cleaning the data, but since this is not the primary concern of our study, we opt to omit it for simplicity.

\end{document}